\documentclass{cccg18}
\usepackage{graphicx,amssymb,amsmath}
\usepackage[titlenumbered,linesnumbered,ruled,noend,algosection]{algorithm2e}

\newtheorem{defn}[theorem]{Definition}
\newtheorem{alg}{\bf{Algorithm}}

\title{Looking for Bird Nests: Identifying Stay Points with Bounded Gaps}

\author{Ali Gholami Rudi\thanks{Department of Electrical and Computer Engineering,
	Bobol Noshirvani University of Technology, {\tt gholamirudi@nit.ac.ir}}}

\index{Gholami Rudi, Ali}

\begin{document}
\thispagestyle{empty}
\maketitle

\begin{abstract}
A stay point of a moving entity is a region in which
it spends a significant amount of time.
In this paper, we identify all stay points of an entity in
a certain time interval, where the entity
is allowed to leave the region but it should return within
a given time limit.
This definition of stay points seems more natural in many
applications of trajectory analysis than those
that do not limit the time of entity's absence from the region.
We present an $O(n \log n)$ algorithm for trajectories in $R^1$
with $n$ vertices and a $(1 + \epsilon)$-approximation algorithm
for trajectories in $R^2$ for identifying all such stay points.
Our algorithm runs in $O(kn^2)$, where $k$ depends on $\epsilon$
and the ratio of the duration of the trajectory to the allowed gap time.
\end{abstract}

\section{Introduction}

The question, asking where a moving entity, like an animal or
a vehicle, spends a significant
amount of its time is very common in trajectory analysis \cite{zheng15}.
These regions are usually called popular places,
hotspots, interesting places, stops, or stay points in the literature.
There are several definitions of stay points and different
techniques have been presented to find them \cite{benkert10,gudmundsson13,fort14,perez16,arboleda17}.
However, from a geometric perspective, which is the focus
of the present paper, few papers are dedicated to this problem.

Benkert et al.~\cite{benkert10} defined a popular place to be
an axis-aligned square of fixed side length in the plane which is visited
by the most number of distinct trajectories.  They modelled a
visit either as the inclusion of a trajectory vertex or
the inclusion of any portion of a trajectory edge, and
presented optimal algorithms for both cases.
Gudmundsson et al.~\cite{gudmundsson13} introduced several
different definitions of trajectory hotspots.  In some of
these definitions, a hotspot is an axis-aligned square that contains
a contiguous sub-trajectory with the maximum duration and
in others it is an axis-aligned square in which the entity spends
the maximum possible duration but its presence may not be contiguous.
For hotspots of fixed side length,
for the former they presented an $O(n \log n)$ algorithm
and for the latter they presented an algorithm with the
time complexity $O(n^2)$, where $n$ is the number of trajectory
vertices.
Damiani et al.~\cite{damiani14}, like some of the cases
considered by Gudmundsson et al.~\cite{gudmundsson13},
allowed gaps between stay point and presented
heuristic algorithms for finding them.

There are applications in which we need to identify regions
that are regularly visited.
Djordjevic et al.~\cite{djordjevic11} concentrated on a
limited form of this problem and presented an algorithm
to decide if a region is visited almost regularly (in fixed
periods of time) by an entity.
However, in many applications that require spatio-temporal
analysis, these definitions are inadequate.
For instance, a bird needs to return to its nest regularly
to feed its chicks.  In other words, the bird may leave
its nest but it cannot be away for a long time.
We would like to find all possible locations for its nest.

Arboleda et al.~\cite{arboleda17} studied a problem very similar to
the focus of the present paper, except that they assumed the algorithm
takes as input, in addition to the trajectories, a set of
polygons as potential stay points or interesting sites.
They presented a simple algorithm to identify stay points among the
given interesting sites; their algorithm computes the longest
sub-trajectory visiting each interesting site for each trajectory,
while allowing the entity to leave the site for some predefined
amount of time.
They also mentioned motivating real world examples to show that in some
applications, it makes sense to allow the entity to leave the site for
short periods of time, like leaving a cinema for the bathroom.

Our goal is identifying all trajectory stay points, i.e.~axis-aligned
squares in which the entity is always present,
except for short periods of time,
where both the side length of the squares and the allowed gap time are
specified as inputs of the algorithm and assumed to be fixed.
Note that we ignore the duration in which the entity stays in a region.
If, for instance, a region with the maximum duration among
our stay points is desired, our algorithm can be combined with
those that find a stay point with the maximum duration,
but allow unbounded entity absence,
like the ones presented by Gudmundsson et al.~\cite{gudmundsson13}.

This paper is organized as follows.
In Section~\ref{sprel}, we introduce the notation and
define some of the main concepts of this paper.
In Section~\ref{sone}, we handle trajectories in $R^1$
and present an algorithm to find all stay points of
such trajectories with the time complexity $O(n \log n)$.
We focus on trajectories in $R^2$ in Section~\ref{stwo}
and present an approximation algorithm for finding their
stay points.  We conclude this paper by showing that the
complexity of the stay map of two-dimensional trajectories
can be $\Theta(n^2)$.

\section{Preliminaries}
\label{sprel}
A trajectory $T$ describes the movement of an entity in a certain
time interval.
Trajectories can be modelled as a set of vertices and edges in
the plane.
Each vertex of $T$ represents a location at which the entity was observed.
The time of this observation is indicated as the timestamp of the vertex.
We assume that the entity moves in a straight line and with
constant speed from a vertex to the next; the edges of the
trajectory connect its contiguous vertices.
A sub-trajectory of $T$ for a time interval $(a, b)$ is
denoted as $T(a, b)$, and describes the movement of the entity
from time $a$ to time $b$.  Except possibly the first and the
last vertices of a sub-trajectory, which may fall on an edge of $T$,
its set of vertices is a subset of those of $T$.
The stay points considered in this paper are formally described
in Definition~\ref{dstaypoint}.  We use the symbols defined
here, such as $g$ and $s$, throughout the paper without
repeating their description.  Also, any square that appears in
the rest of this paper is axis-aligned and has side length $s$.
\begin{figure}
	\centering
	\includegraphics[width=\linewidth]{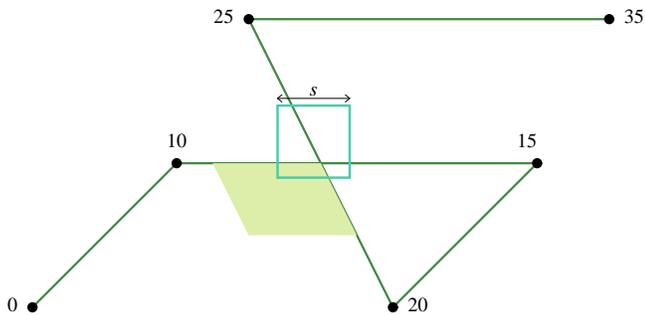}
	\caption{An example two-dimensional trajectory.
	The number near each vertex shows its timestamp.
	The green region is the stay map
	and the green square is a stay point ($g = 15$).}
	\label{fsmap}
\end{figure}

\begin{defn}
\label{dstaypoint}
A stay point of a trajectory $T$ in $R^2$ is a square of fixed
side length $s$ in the plane such that the entity never spends
more than a given time limit $g$ outside it continuously.
\end{defn}

The goal of this paper is identifying all stay points of a trajectory,
or its stay map (Definition~\ref{dstaymap}).
Note that the parameters $s$ and $g$ are assumed to be fixed and
specified as inputs of the algorithm.

\begin{defn}
\label{dstaymap}
The stay map $M$ of a trajectory $T$ in $R^2$ is a subset of
the plane such that every square of side length $s$ whose lower
left corner is in $M$ is a stay point of $T$, and the lower left
corners of all stay points of $T$ are in $M$.
\end{defn}

Figure~\ref{fsmap} shows an example trajectory, its stay map,
and one of its stay points.  Note that every square, whose lowest
left corner is in the stay map, is a stay point.
Although these definitions are presented for trajectories in $R^2$,
they can be trivially adapted for one-dimensional trajectories,
as we do in Section~\ref{sone}.

\section{Stay Maps of One-Dimensional Trajectories}
\label{sone}
Let $T$ be a trajectory in $R^1$.
A stay point of $T$ is an interval of length $s$
such that the entity never leaves it for a period of time
longer than $g$.
The stay map $M$ of $T$ is the region containing the left end points
of all stay points of $T$.
In this section, we present an algorithm for finding $M$.
\begin{figure}
	\centering
	\includegraphics[width=\columnwidth]{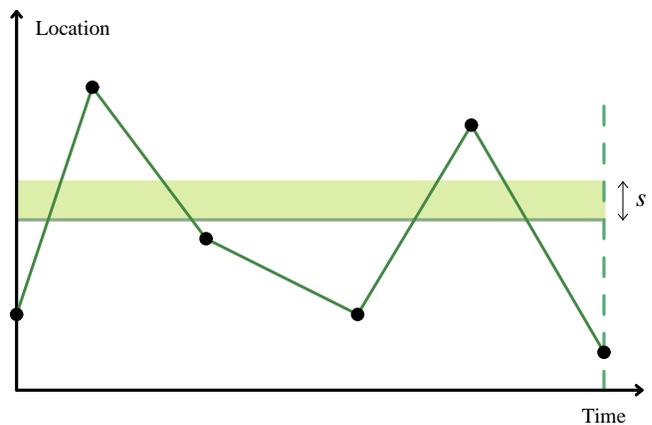}
	\caption{Mapping a one-dimensional trajectory to
	the time-location plane.
	The green rectangle of height $s$ shows a possible stay point.}
	\label{fonemap}
\end{figure}

\begin{lemma}
\label{lonecon}
The stay map $M$ of a trajectory $T$ in $R^1$ is continuous.
\end{lemma}
\begin{proof}
To obtain a contradiction, let points $p$ and $q$ be inside $M$ and
$v$ be outside it such that $p < v < q$ (our assumption that $M$ is
non-continuous implies the existence of this triple).
Let $r_p$, $r_q$, and $r_v$ be three segments of length $s$,
whose left corners are at $p$, $q$, and $v$, respectively.
Clearly, $r_p$ and $r_q$ are stay points while $r_v$ is not.
Whenever the entity moves to the left of $v$, it must return to $q$
before the time limit $g$ to visit $r_q$.
Also, whenever the entity moves beyond the right
end point of $r_v$ (which is outside $r_p$), it must return
to $r_p$ before the time limit.  Therefore, it can never be
outside $r_v$ for more than time $g$ and this implies that $v$
is also a stay point and inside $M$, which yields the desired
contradiction.
\end{proof}

\begin{lemma}
\label{lonecheck}
Given a trajectory $T$ with $n$ vertices in $R^1$, we can answer in $O(n)$
time whether a point $p$ is in the stay map or not, and if not, whether
the stay map is on its left side or on its right side.
\end{lemma}
\begin{proof}
Define $r$ as segment $pq$, in which $q$ is $p + s$.
Testing each trajectory edge in order, we can compute the duration
of each maximal sub-trajectory outside $r$ and check if it is at most $g$.
Therefore, we can decide if $p$ is the left end point of a
stay point in $O(n)$ time.  If it is not a stay point,
there is at least one time interval, in which the entity spends
more than time $g$ on the left or on the right side of $r$.
Without loss of generality, suppose it does so on the left side.
Then, no point on the right of $r$ can be a stay point and
therefore the whole stay map of $T$ must appear on the left of $p$.
This again can be tested in $O(n)$ time by processing trajectory
edges.
\end{proof}

An event point of a trajectory $T$ in $R^1$ is a point on the
line in which one of the following occurs:
i) a trajectory vertex lies on that point,
ii) the time gap between two contiguous visits to that
point is exactly $g$.

\begin{lemma}
\label{loneends}
The stay map $M$ of a trajectory $T$ starts and ends at an event point
or at distance $s$ from one.
\end{lemma}
\begin{proof}
By Lemma~\ref{lonecon}, $M$ is continuous.
Let $p$ be the left end point of the stay map $M$.
Let $r = pq$ be a segment such that $q = p + s$.
Whenever the entity leaves $r$ through $p$, it returns by
passing it again within the time limit $g$.  Similarly, if the
entity leaves $r$ through $q$, it visits $q$ again within
time $g$.  Suppose, for the sake of contradiction, that
$p$ is not an event point.  Then, we can move $r$ slightly
to the left to obtain $r'$.  $r'$ must also be a stay point
because every time the entity leaves it from either of its end points,
it returns within time $g$, because neither $p$ nor $q$ is an event
point (the time between the contiguous visits of the entity is
not exactly $g$ and they are not on a trajectory vertex).
This contradicts the choice of $p$.
A similar argument shows that the right endpoint of $M$ must
also be an event point or at distance $s$ from one.
\end{proof}

\begin{lemma}
\label{toneevents}
The set of event points of a trajectory with $n$ vertices
can be computed in $O(n \log n)$ time.
\end{lemma}
\begin{proof}
We map the trajectory to a plane such that a trajectory
vertex at position $p$ with timestamp $t$ is mapped to point $(t, p)$
(see Figure~\ref{fonemap}).
Obviously, the polygonal path representing the trajectory in this
plane is $y$-monotone.
We perform a plane sweep by sweeping a line parallel to the $x$-axis
in the positive direction of the $y$-axis in this plane.

The edges in this plane chop the sweep line into several
segments.  We maintain the length of every such segment during
the sweep line algorithm.  When the sweep line intersects a
trajectory vertex $v$, an event point is recorded and,
based on the other end point of the edges that meet
at that vertex, one of the following cases occurs:
\begin{enumerate}
\item
If $v$ is the lowest end point of both edges,
two new segments are introduced.
Based on the slope of the edges bounding each segment,
we record an event at which the distance between the edges
is exactly $g$, if they are long enough.
\item
If $v$ is the highest end point of both edges that meet at $v$,
three segments on the sweep line are merged (when the sweep line
is before $v$, three segments are created by the edges incident
to $v$, at $v$, there are two such segments, and
after $v$, they merge into one).
We also record an event for the location at which the length of
the remaining segment becomes $g$ in the plane.
\item
If $v$ is the highest end point of one edge and the lowest
end point of another,
the event scheduled for the location at which the length of
each of the two incident segments on the sweep line are $g$
may need to be updated.
\end{enumerate}
Note that since the sweep line stops at $n$ vertices and at each
vertex only a constant number of event points are added,
the total number of event points is $O(n)$.
\end{proof}

\begin{theorem}
\label{tonealg}
The stay map $M$ of a trajectory $T$ with $n$ vertices in $R^1$
can be computed in $O(n \log n)$ time.
\end{theorem}
\begin{proof}
Lemma~\ref{toneevents} implies that the set of event points of
$T$ can be computed with the time complexity $O(n \log n)$.
From this set, we can obtain an ordered sequence of event points
and points at distance exactly $s$ from them in $O(n \log n)$ time
(note that the length of this sequence is still $O(n)$).
Based on Lemma~\ref{loneends}, $M$ starts and ends at a point
of this sequence.
Also, Lemma~\ref{lonecheck} implies that we can decide if any of the
end points of $M$ appears before or after any point in $O(n)$ time.
Therefore, we can perform a binary search on the sequence obtained
from the event points of $T$ to find the left and the right end points of $M$.
Since the length of the sequence is $O(n)$, the time complexity
of the binary search is $O(n \log n)$.
\end{proof}

Unfortunately, this algorithm cannot be adapted for two-dimensional
trajectories, because their stay maps may no longer be continuous.

\section{Stay Maps of Two-Dimensional Trajectories}
\label{stwo}
We use the notation $P(a, b)$ to denote the region that contains
the lower left corners of all squares of side length $s$ that
contain at least one point of the sub-trajectory $T(a, b)$.
We also use $M(a, b)$ to indicate the stay map of the sub-trajectory $T(a, b)$.
We assume that trajectory $T$ starts at time $0$ and
has total duration $D$.
It is clear that every point in the stay
map of $T$ must appear in $P(t, t + g)$ for any value
of $t$, where $0 \le t \le D - g$ (because the entity cannot be
outside a stay point of $T$ for more than time $g$).
Therefore, the stay map of $T$ is the intersection of $P(t, t + g)$
for every possible value of $t$, $0 \le t \le D - g$.
This suggests the general scheme demonstrated in Algorithm~\ref{atwoexact}
for finding the stay map of a two-dimensional trajectory, assuming $D > g$.

\begin{alg}
\label{atwoexact}
Let $T$ be two-dimensional trajectory with $n$ edges and total duration $D$.
Compute the stay map of $T$ ($M(0, D)$) as follows.
\begin{enumerate}
\item Compute $P(0, g)$, as the union of polygons $P(u, v)$,
for all edges $uv$ in $T(0, g)$.
\item Let $M(0, g)$ be $P(0, g)$.  This is not strictly correct as
$M(0, t)$ must include the complete plane when $t \le g$ and its
value changes to a subset of $T(0, g)$ for any value of $t > g$.
This simplifying assumption, however, does not affect the
correctness of the algorithm, since $D > g$.
\item Incrementally compute $M(0, D)$ as follows.
Compute $M(0, b)$ from $M(0, a)$, in which
$M(0, a)$ is the last computed stay map and
$b$ is the smallest value after $a$, such that $b - g$
or $b$ is the timestamp of a trajectory vertex.
Let $V$ be the difference between $M(0, a)$ and $M(0, b)$
(note again that $M(0, b)$ is a subset of $M(0, a)$).
After computing $V$, we obtain $M(0, b)$ by excluding
$V$ from $M(0, a)$.
\end{enumerate}
\end{alg}

\begin{figure}
	\centering
	\includegraphics[width=\linewidth]{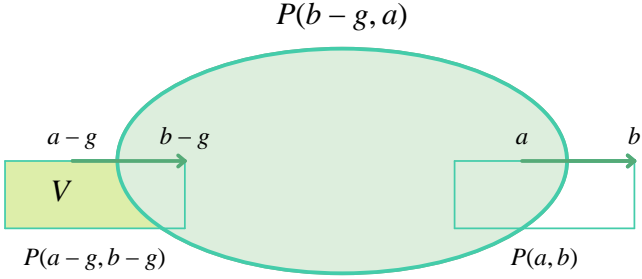}
	\caption{The difference $V$ in Algorithm~\ref{atwoexact},
	when $P(a - g, b - g)$ and $P(a, b)$ do not overlap.}
	\label{fvcomp}
\end{figure}
The core of Algorithm~\ref{atwoexact} is the computation of the
difference $V$.
By the choice of $b$, $T(a - g, b - g)$ and $T(a, b)$ are both line
segments.  The value of $V$ depends on these segments and $T(b - g, a)$.

Let $r$ be a square, whose lower left corner is in $V$ and
let $a - g + \delta$ be the time of entity's departure from
$r$ before time $b - g$.  Since the lower left corner
of $r$ is in $V$, $r$ is not visited by the entity in
the sub-trajectory $T(a - g + \delta, a + \delta)$.
In other words, any point not in $P(a - g + \delta, b - g)$,
$P(b - g, a)$, and $P(a, a + \delta)$ for any value of $\delta$
in $0 \le \delta \le g$ cannot be a stay point.

To make the computation of $V$ easier, we define $V'$ as
follows ($V'$ is very similar to $V$, except that it ignores $P(b - g, a)$):
{
\scriptsize
\[
V' = \bigcup\limits_{0 \le \delta \le g} {P(a - g, a - g + \delta) \setminus
\left( P(a - g + \delta, b - g) \cup P(a, a + \delta)\right)}
\]
}
$V'$ contains the lower left corners of all squares that have been
visited during the interval $(a - g, a - g + \delta)$,
but have not been visited in
$(a - g + \delta, b - g)$ or $(a, a + \delta)$
for some $\delta$ in $0 \le \delta \le g$.
Then, $V = V' \setminus P(b - g, a)$.

If no square intersects both $T(a - g, b - g)$ and $T(a, b)$,
$V'$ is $P(a - g, b - g)$.  This case is shown in Figure~\ref{fvcomp},
in which $V'$ is the rectangle on the left.
Otherwise, $V'$ depends on the relative
speed of the entity in these sub-trajectories.
In both cases, $V'$ is a polygon of constant complexity and can be computed
in constant time.
We do not discuss the details of the computation
of $V'$ in this paper, however.
Since $T(b - g, a)$ consists of $O(n)$ edges,
$P(b - g, a)$ is the union of $O(n)$ simple polygons.
Therefore, $V' \setminus P(b - g, a)$ is also the union of a
set of polygons with the total complexity $O(n)$.
Let $V_t$ be the union of the differences $V$ for all
iterations of the third step of Algorithm~\ref{atwoexact}
(note that the complexity of $V_t$ is $O(n^2)$).
When the algorithm finishes, $M(0, D)$ is $P(0, g) \setminus V_t$.
Since the computation of $V_t$ requires finding the union of polygons
with the total complexity $O(n^2)$, an $O(n^2)$ implementation
of this exact algorithm seems unlikely.

\subsection{Approximate Stay Maps of Two-Dimensional Trajectories}
In Algorithm~\ref{atwo}, we consider $P(t, t + g)$ for limited discrete
values of $t$ to compute \emph{approximate stay maps} of
a trajectory (Definitions~\ref{dstaypointapx} and \ref{dstaymapapx}),
to improve the time complexity of Algorithm~\ref{atwoexact}.

\begin{defn}
\label{dstaypointapx}
A $(1 + \epsilon)$-approximate stay point of a trajectory $T$ in $R^2$
is a square of fixed side length $s$,
such that the entity is never outside it for more than $g + \epsilon g$ time.
\end{defn}

\begin{defn}
\label{dstaymapapx}
A $(1 + \epsilon)$-approximate stay map of a trajectory $T$ in $R^2$
is the region containing the lower left corners of all exact stay points
of $T$ and possibly the lower left corners of some of its
$(1 + \epsilon)$-approximate stay points.
\end{defn}

\begin{alg}
\label{atwo}
Let $T$ be a trajectory in $R^2$ with $n$ edges and total
duration $D$ and let $\epsilon$ be any real positive constant no greater than $D / g$.
Compute a $(1 + \epsilon)$-approximate stay map of $T$ as follows.
\begin{enumerate}
\item Compute $P(t, t + g)$ for $t = i\lambda$ for integral values of $i$
from $0$ to $D / \lambda$, where $\lambda$ is $\epsilon g$.
We call $P(t, t + g)$ for any value of $t$ a snapshot of $T$.
\item Compute the intersection of these snapshots.  For this, we can
use the topological sweep of Chazelle and Edelsbrunner~\cite{chazelle92}
on the subdivision of the plane induced by the edges of the
snapshots and include in the output the faces present in all snapshots.
\end{enumerate}
\end{alg}

\begin{theorem}
\label{ttwo}
For trajectory $T$ in $R^2$ with $n$ edges and total
duration $D$ and any real positive constant $\epsilon$ no greater than $D / g$,
Algorithm~\ref{atwo} computes a $(1 + \epsilon)$-approximate stay map of $T$.
\end{theorem}
\begin{proof}
Since the output of Algorithm~\ref{atwo} is the intersection of
different snapshots of $T$, the lower left corner of every stay point
must be inside it.  Therefore, it suffices to show that every point
in the output of the algorithm is the lower left corner of a
$(1 + \epsilon)$-approximate stay point.

Let $r$ be a square whose lower left corner is in the region
reported by this algorithm.
Suppose that the entity leaves $r$ at $t_b$ and reenters $r$ at $t_e$.
We can set $t_b = 0$ for handling the initial part of the trajectory,
and, if the entity never returns to $r$, we can set $t_e = D$.
To prove the approximation factor, we show that $t_e \le t_b + g + \epsilon g$.
Let $i$ be the largest index such that $\lambda i \le t_b$
and let $t_1 = \lambda i$.
We show that the entity must return before time
$t_1 + \lambda + \lambda/\epsilon$.
Otherwise, $P(t_1 + \lambda, t_1 + \lambda + \lambda/\epsilon)$,
which is a snapshot since $\lambda/\epsilon$ is equal to $g$,
does not contain the
lower left corner of $r$ (this is demonstrated in Figure~\ref{fsnap})
and this contradicts the assumption
that it is included in the region returned by the algorithm.
Therefore, the entity cannot be outside $r$ for longer than
$\lambda/\epsilon + \lambda$, and $t_e \le t_b + g + \epsilon g$.
\end{proof}

\begin{figure}
	\centering
	\includegraphics[width=\linewidth]{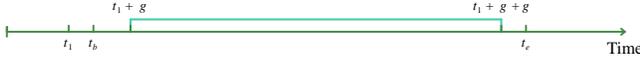}
	\caption{The entity leaves a square at $t_b$ and returns at $t_e$.
	If $t_e - t_b$ is larger than $g + g\epsilon$, there is a
	snapshot in which the entity is outside the square.}
	\label{fsnap}
\end{figure}
\begin{theorem}
\label{ttwoanalysis}
The time complexity of Algorithm~\ref{atwo} is $O(n^2 / \epsilon^2 + \sigma^2 / \epsilon^2)$,
in which $\sigma$ is $D / g$.
\end{theorem}
\begin{proof}
A subdivision of the plane by $m$ line segments
has $O(m^2)$ faces and can be swept with the same
time complexity \cite{chazelle92}.
Moreover, the number of the segments of each snapshot
depends on the number of vertices of the sub-trajectory
inside that snapshot (the region containing the lower left
corners of the squares that intersect an edge of the sub-trajectory
is a polygon with a constant number of sides).
We, therefore, count the total number of vertices
of the sub-trajectories in all snapshots.  There are two types
of trajectory vertices in each snapshot: those present in the
original trajectory $T$ and the end points of the snapshot,
which may not coincide with a trajectory vertex.
Since the duration of each snapshot is $g$ and the difference
between the start time of contiguous snapshots is $\epsilon g$,
each trajectory vertex appears in at most $1 / \epsilon$ snapshots.
Therefore, the total number of vertices is at most $n / \epsilon + 2D / (\epsilon g)$
and the time complexity of Algorithm~\ref{atwo} is
$O(n^2 / \epsilon^2 + \sigma^2 / \epsilon^2)$.
\end{proof}

It is not difficult to see that the stay map of a two-dimensional
trajectory may contain $\Theta(n^2)$ faces and therefore we cannot
expect an algorithm with the worst-case time complexity $o(n^2)$.
In what follows, we demonstrate a trajectory with $O(n)$ edges and
a stay map of $\Theta(n^2)$ faces.
Trajectory edges are added incrementally, as demonstrated in
Figure~\ref{fgrid}, in which filled regions represent the stay map
(except for $t \le g$, in which they represent $P(0, t)$) and
arrows show trajectory edges.  We assume that the entity starts at
time $0$ and position $(0, 0)$.

\begin{figure*}
	\centering
	\includegraphics[width=\linewidth]{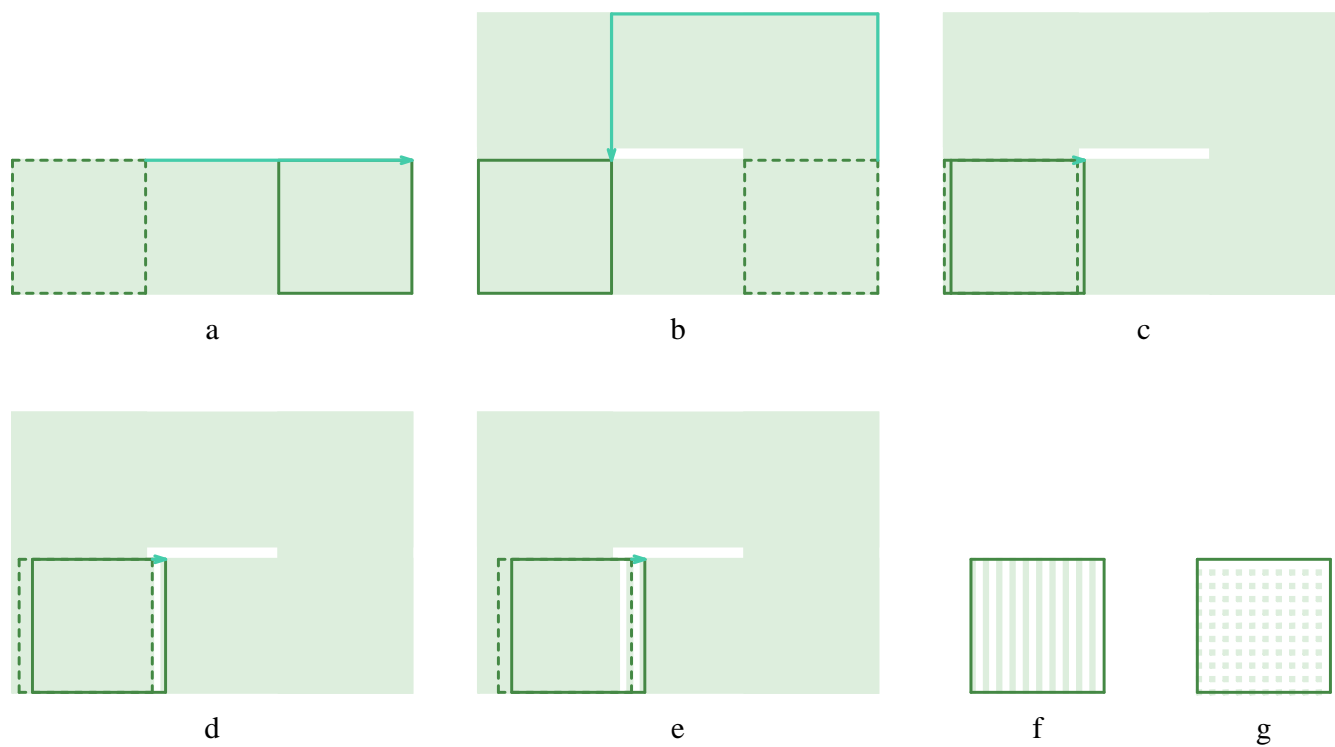}
	\caption{A trajectory with a stay map of $O(n^2)$ faces.
	The arrows indicate trajectory edges and filled
	regions indicate the stay map at each step.}
	\label{fgrid}
\end{figure*}

Generate $m$ vertical strips as follows.
Add the second vertex at $(2s, 0)$ with timestamp $g/2$ (Figure~\ref{fgrid}.a).
Move the entity to its initial position using three vertices as shown in Figure~\ref{fgrid}.b;
the position of the last vertex is $(0, 0)$ and its timestamp is $g - g / 2m$.
Create the vertical strips as follows:
after every $g / m$ time, quickly move the
entity by $s / m$ to the right (Figures~\ref{fgrid}.c--\ref{fgrid}.e).
After $m$ such steps and waiting for at least $g$,
the current stay map consists of $m$
vertical strips (Figure~\ref{fgrid}.f).

The same trajectory we used for creating vertical strip can be used
for creating horizontal strips after rotating the trajectory 90
degrees.  If this is performed after the previous step, however, this
would result in a stay map (Figure~\ref{fgrid}.g), which
consists of $\Theta(m^2)$ small squares.

\section{Concluding Remarks}
The definition of stay points with bounded gaps can be easily
extended to multiple trajectories.  A multi-trajectory stay
point is a square that is visited by at least one of the entities
in any interval of duration $g$.  It seems possible to compute
such stay maps, by modifying Algorithm~\ref{atwo} to compute
the intersection of the union of the snapshots of different entities.
However, the time complexity of this algorithm may no longer be $O(n^2)$,
where $n$ is the total number of trajectory vertices.
Finding an efficient exact algorithm for the multi-trajectory
version of the problem seems interesting.

As shown in Section~\ref{stwo}, the complexity of a stay map can be
$\Theta(n^2)$, rendering an algorithm with the time complexity
$o(n^2)$ impossible.  This bound however is not tight and a natural
question is whether it is possible to find the exact stay map of
two-dimensional trajectories in $O(n^2)$ time.  Also, by limiting the size
of the output, for instance by finding only one of the stay points, a
more efficient algorithm is not unlikely.
Furthermore, it seems interesting to study the problem in higher
dimensions.

\section*{Acknowledgement}
We thank Neda Ahmadzadeh Tori for the inspiring discussions
that led to the study of this problem.
\small
\bibliographystyle{abbrv}

\begin{thebibliography}{11}

\bibitem{zheng15}
Y.~Zheng.
\newblock Trajectory data mining - an overview.
\newblock {\em ACM Transactions on Intelligent Systems and Technology},
  6(3):29:1--29:41, 2015.

\bibitem{benkert10}
M.~Benkert, B.~Djordjevic, J.~Gudmundsson, and T.~Wolle.
\newblock Finding popular places.
\newblock {\em International Journal of Computational Geometry and
  Applications}, 20(1):19--42, 2010.

\bibitem{gudmundsson13}
J.~Gudmundsson, M.~J. van Kreveld, and F.~Staals.
\newblock Algorithms for hotspot computation on trajectory data.
\newblock In {\em SIGSPATIAL/GIS}, pages 134--143, 2013.

\bibitem{fort14}
M.~Fort, J.~A. Sellar{\`e}s, and N.~Valladares.
\newblock Computing and visualizing popular places.
\newblock {\em Knowledge and Information Systems}, 40(2):411--437, 2014.

\bibitem{perez16}
R.~Pérez-Torres, C.~Torres-Huitzil, and H.~Galeana-Zapién.
\newblock Full on-device stay points detection in smartphones for
  location-based mobile applications.
\newblock {\em Sensors}, 16(10):1693, 2016.

\bibitem{arboleda17}
F.~J.~M. Arboleda, V.~Bogorny, and H.~Patiño.
\newblock Smot+ncs - algorithm for detecting non-continuous stops.
\newblock {\em Computing and Informatics}, 3(2):283--306, 2017.

\bibitem{damiani14}
M.~L. Damiani, H.~Issa, and F.~Cagnacci.
\newblock Extracting stay regions with uncertain boundaries from gps
  trajectories - a case study in animal ecology.
\newblock In {\em SIGSPATIAL/GIS}, pages 253--262, 2014.

\bibitem{djordjevic11}
B.~Djordjevic, J.~Gudmundsson, A.~Pham, and T.~Wolle.
\newblock Detecting regular visit patterns.
\newblock {\em Algorithmica}, 60(4):829--852, 2011.

\bibitem{chazelle92}
B.~Chazelle and H.~Edelsbrunner.
\newblock An optimal algorithm for intersecting line segments in the plane.
\newblock {\em Journal of the ACM}, 39(1):1--54, 1992.

\end{thebibliography}

\end{document}